\documentclass[a4paper,USenglish]{lipics-v2016-arxiv}
\usepackage{microtype}%if unwanted, comment out or use option "draft"
\usepackage{paralist}
\usepackage{hyperref} %does not work with lipics?
\usepackage{amsmath,amssymb,amsthm}%  AMS
\usepackage{graphicx} 
\theoremstyle{plain}
\newtheorem{proposition}[theorem]{Proposition}
\let\phi=\varphi

\title{Loopless Gray Code Enumeration and the Tower of Bucharest}
\author[1]{Felix Herter}
\author[1]{G\"unter Rote}
\affil[1]{Institut f\"{u}r Informatik, Freie Universit\"{a}t Berlin\\
  Takustr. 9, 14195 Berlin, Germany\\
  \texttt{avealx@zedat.fu-berlin.de},
  \texttt{rote@inf.fu-berlin.de}}
\authorrunning{F. Herter and G. Rote} %mandatory. First: Use abbreviated first/middle names. Second (only in severe cases): Use first author plus 'et. al.'
\Copyright{Felix Herter and G\"{u}nter Rote}%mandatory, please use full first names. LIPIcs license is "CC-BY";  http://creativecommons.org/licenses/by/3.0/

\EventLongTitle{This paper is to appear without the full appendix in the 
\emph{8th International Conference on Fun with Algorithms} (FUN 2016)
 in June 2016;
Editors: {Erik D. Demaine and Fabrizio Grandoni},
Leibniz International Proceedings in Informatics.
doi:\href{http://dx.doi.org/10.4230/LIPIcs.FUN.2016.19}{10.4230/LIPIcs.FUN.2016.19}}
\EventNoEds{2}
\EventShortTitle{FUN 2016}
\EventAcronym{FUN}
\EventYear{2016}
\EventDate{June 8--10, 2016}
\EventLocation{La Maddalena, Italy}
\EventLogo{}
\SeriesVolume{49}
\ArticleNo{19} % Your paper number goes here!
\begin{document}

\maketitle

\begin{abstract}
  We give new algorithms for generating all $n$-tuples over
  an alphabet of $m$ letters, changing only one letter at a time (Gray
  codes).  These algorithms are based on the connection with %certain
  variations of the Towers of Hanoi game.  Our algorithms are
  loopless, in the sense that the next change can be determined in a
  constant number of steps, and they can be implemented in hardware.
  We also give another family of loopless algorithms that is based on
  the idea of working ahead and saving the work in a buffer.

\end{abstract}

\tableofcontents\markboth{Felix Herter and G\"{u}nter Rote}{Loopless Gray Code Enumeration and the Tower of Bucharest}

%\section{Introduction: The Reflected Gray Code and the Towers of Hanoi}
%\label{sec:intro2}
\section{Introduction: the binary reflected Gray code and
 the %!!!
 Towers of Hanoi}
\label{sec:intro}

\subsection{The Gray code}
%The reflected
The Gray code, or more precisely,
the reflected binary Gray code $G_n$,
 orders the $2^n$ binary strings of length $n$
in such a way that successive strings differ in a single bit.
%is a particular Gray code, which 
It is
defined inductively as follows, see Figure~\ref{gray4} for an example.
The Gray code  $G_1=0,1$, and if
$G_n=C_1,C_2,\ldots,C_{2^{n}}$ is the Gray code for the bit strings of
length $n$, then % the strings of length $n+1$ are generated in the
\begin{gather}
\label{def-gray}
G_{n+1} = 0C_1,0C_2,\ldots ,0C_{2^{n}},\ 
1C_{2^{n}},
1C_{2^{n}-1},\ldots,1C_2,1C_1
.
\end{gather}
In other words, %This recursive rule means that 
we prefix each word of $G_n$ with 0, and
this is followed by the reverse of $G_n$ with 1 prefixed to each word.
\begin{figure}[htb]
\centering
$
\ifvoid0\else\setbox2=\box0\fi
\def\mu#1 {\count2=#1 \count0=#1 \mau}
\def\mau#1 {%
\ifx\relax#1%
 \let\next=\relax
 \ifvoid 0\else
   \box0 
 \fi
\else
 \let\next=\mau
 \ifnum\count0=0
    \box0 
    { \vrule\ }%
    \count0=\count2
 \fi   
 \advance \count0 by -1
 \ifvoid 0 \setbox0=\vbox{\vskip 4pt}%
 \else
%    \setbox0=\vtop{\unvbox0\vskip 4pt}%
 \fi
 \setbox0=\vtop{\unvbox0\hbox{#1}\vskip 4pt}%
\fi
\next}%
\vcenter{
\vskip 2pt
\hbox{\mu 13
000000
000001
000011
000010
000110
000111
000101
000100
001100
001101
001111
001110
001010
001011
001001
001000
011000
011001
011011
011010
011110
011111
011101
011100
010100
010101
010111
010110
010010
010011
010001
010000
110000
110001
110011
110010
110110
110111
110101
110100
111100
111101
111111
111110
111010
111011
111001
111000
101000
101001
101011
101010
101110
101111
101101
101100
100100
100101
100111
100110
100010
100011
100001
100000
{\relax}
}\vskip 2pt}
\quad\vrule\quad
\vcenter{
\hbox{\mu 13
0000
0001
0002
0012
0011
0010
0020
0021
0022
0122
0121
0120
0110
0111
0112
0102
0101
0100
0200
0201
0202
0212
0211
0210
0220
0221
0222
1222
1221
1220
1210
1211
1212
1202
1201
1200
1100
1101
1102
1112
1111
1110
1120
1121
1122
1022
1021
1020
1010
1011
1012
1002
1001
1000
2000
2001
2002
2012
2011
2010
2020
2021
2022
2122
2121
2120
2110
2111
2112
2102
2101
2100
2200
2201
2202
2212
2211
2210
2220
2221
2222
{\relax}
}}$
  \caption{The binary Gray code $G_6$ for 6-tuples
and the ternary Gray code for 4-tuples.}
  \label{gray4}
\end{figure}

\subsection{Loopless algorithms}
The Gray code has an advantage 
 over alternative algorithms for
enumerating the binary strings, for example in lexicographic order:
 one can change a binary string $a_na_{n-1}\ldots a_1$ to the successor
in the sequence by a single update of the form $a_i := 1-a_i$ in
constant time.
However, we also have to \emph{compute} the position $i$ of the bit
which has to be updated.
A~straightforward implementation of the recursive definition
\eqref{def-gray}
 leads to an algorithm with an optimal overall runtime of
$O(2^n)$, i.e., constant average time per enumerated bit string.

A stricter requirement is to compute each successor string in
constant \emph{worst-case} time. Such an algorithm is called a
\emph{loopless} generation algorithm.
Loopless enumeration algorithms for various combinatorial structures were pioneered by Ehrlich~\cite{Ehrlich},
and different loopless algorithms for Gray codes are known, see
Bitner, Ehrlich, and Reingold~\cite{Bitner} and Knuth
\cite[Algorithms 7.2.1.1.L and 7.2.1.1.H]{kn4}.
These algorithms achieve
constant running time by maintaining additional pointers. % in a smart way.
%\cite[Algorithm~7.2.1.1H, p.~300]{kn4} % Gray-code

\subsection{The Tower of Hanoi}

The Tower of Hanoi is the standard textbook example for illustrating
the principle of recursive algorithms.  It has $n$ disks
$D_1,D_2,\dots,D_{n}$ of increasing radii and three pegs
$P_0,P_1,P_2$, see Fig.~\ref{hanoi}.
The goal is to move all disks from the peg $P_0$, where they initially
rest, to another peg, subject to the following rules:
\begin{enumerate}
\item Only one disk may be moved at a time:
the topmost disk from one peg can be moved on top of the disks of another peg
\item A disk can never lie on top of a smaller disk.
\end{enumerate}

\begin{figure}[htb]
  \centering
 \includegraphics[width=\textwidth]{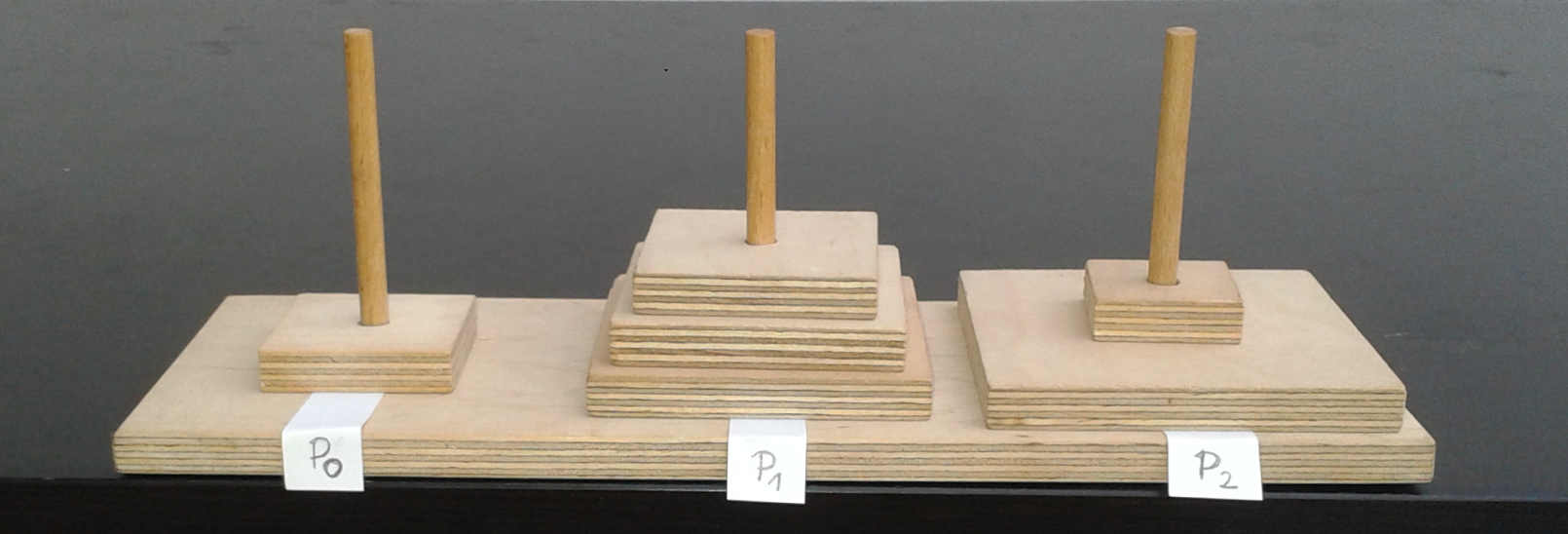}
% %211101 P0:2 P1:5,4,3,1 P2:6
% %211100 P0:2,1 P1:5,4,3 P2:6
% %211200 P0:2,1 P1:5,4 P2:6,3

% 211102 P0:2 P1:5,4,3 P2:6,1

% 211101 P0:2 P1:5,4,3,1 P2:6

% 211100 P0:2,1 P1:5,4,3 P2:6

% 211200 P0:2,1 P1:5,4 P2:6,3

% 211201 P0:2 P1:5,4,1 P2:6,3

% 110011 P0:2 P1:5,4,3 P2:6,1

% 110010 P0:2,1 P1:5,4,3 P2:6

% 110110 P0:2,1 P1:5,4 P2:6,3

% 110111 P0:2 P1:5,4,1 P2:6,3

  \caption{The Towers of Hanoi with $n=6$ (square) disks.
When running the algorithm HANOI from Section~\ref{sec:binary},
the configuration in this picture occurs together with
the bit string 110011.
(There is no easy relation between the positions of the disks and this bit string.)
 The next disk to move is $D_1$; it moves
clockwise to peg $P_0$, and the last bit is complemented.
The successor in the Gray code is
the string
110010. % P0:2,1 P1:5,4,3 P2:6
After that, $D_1$ pauses for one step, while
 disk $D_3$ moves, again clockwise, from $P_1$ to $P_2$, and
the third bit from the right is complemented, leading to
the string
110110. % P0:2,1 P1:5,4 P2:6,3
}
  \label{hanoi}
\end{figure}

For moving a tower of height $n$, one has to move
disk $D_n$ at some point.
But before moving disk $D_n$ from peg $A$ to $B$, one has to move the disks
$D_1,\ldots,D_{n-1}$, which lie on top of $D_n$, out of the way, onto
the third peg. After moving $D_n$ to $B$, these disks have to be moved from the third
peg to $B$. This reduces the problem for a tower of height $n$ to two
towers of height $n-1$, leading to the following recursive procedure.
\begin{tabbing}
  \qquad\=\+
\textit{move\_tower\/}($k,A,B$): (Move the $k$ smallest disks $D_1\ldots D_k$
from peg $A$ to peg $B$)\\
  \qquad\=\+
\textbf{if} $k\le0$: \textbf{return}\\
$\textit{auxiliary} := 3-A-B$; (\textit{auxiliary} is the third peg,
  different from $A$ and $B$.)\\
$\textit{move\_tower\/}(k-1,A,\textit{auxiliary})$\\
move disk $D_k$ from $A$ to $B$\\
$\textit{move\_tower\/}(k-1,\textit{auxiliary},B)$\\
\end{tabbing}

\pagebreak
\subsection{Connections between the Towers of Hanoi and Gray codes}
\label{delta}
The \emph{delta sequence} of the Gray code is the sequence
$1,2,1,3,1,2,1,4,1,2,1,\ldots$ of bit positions that are updated.
(In contrast to the usual convention, we number the bits starting from~1.)
This sequence has an obvious recursive structure which results from
\eqref{def-gray}.
It %This sequence 
also describes the number of changed bits when
incrementing from $i$ to $i+1$ in binary counting.
Moreover,
it is easy to observe that the same sequence also describes the disks that
are moved by the recursive algorithm
\textit{move\_tower} above.
It has thus been noted that the Gray code $G_n$ can be used to solve the %well-known
 Tower of Hanoi puzzle, cf.\
Scorer, Grundy, and Smith~\cite{some-binary} or
Gardner~\cite{gardner}.
In the other direction,
the Tower of Hanoi puzzle can be used to generate the Gray code $G_n$,
see
Buneman and Levy~\cite{buneman-levy}.

%Since there are 
Several loopless ways to compute the next move for the Towers
of Hanoi are known, and they lead directly to loopless algorithms
for the Gray code.
We describe one such algorithm.

\subsection{Loopless Tower of Hanoi and binary Gray code}
\label{sec:binary}
From the recursive algorithm
\textit{move\_tower}, it is not hard to derive the following
fact.
\begin{proposition}
If the tower should be moved from $P_0$ to $P_1$ and
$n$ is odd,
or if the tower should be moved from $P_0$ to $P_2$ and
$n$ is even,
the moves of the odd-numbered disks always proceed in forward (``clockwise'')
circular direction: $P_0\to P_1\to P_2\to P_0$, and
%or ``counterclockwise'': $0\to2\to1\to0$),
the even-numbered disks always proceed in the
opposite % counterclockwise
circular direction: $P_0\to P_2\to P_1\to P_0$.
%or vice versa.
\qed
\end{proposition}
In the other case, when the assumption does not hold, the directions are simply swapped.
Since we are interested not in moving the tower to a specific target
peg, but in generating the Gray code, we stick with the 
proposition as stated.
\begin{tabbing}
    \quad\=\+
\emph{Algorithm} \textbf{HANOI}. Loopless algorithm for the binary
Gray code.\\
    \quad\=\+
Initialize:
Put all disks on $P_0$.\\
    \textbf{loop}:\\
    \qquad\=\+
    Move $D_1$ clockwise.\\ %, and set $a_1 := 1-a_1$.\\
Let $D_k$ be the smaller of the topmost disks on the two pegs that
don't carry $D_1$.
\\
If there is no such disk, terminate.
\\ Move $D_k$ clockwise if $k$ is odd; otherwise, move it
counterclockwise.
%\\ Set $a_k := 1-a_k$.
\end{tabbing}
To obtain the Gray code, we simply set
 $a_k := 1-a_k$ whenever we move the disk $D_k$~\cite{buneman-levy}.
See  Fig.~\ref{hanoi} for a snapshot of the procedure.

We would not need the clockwise/counterclockwise rule for $D_k$: Since
we must not put $D_k$ on top of $D_1$, there is anyway no choice of
where to move it~\cite{buneman-levy}.
We have chosen the above formulation
since it is better suited for generalization.

\subsection{Overview}
In the remainder of this paper, we will generalize these connections
to Gray codes for larger radixes (alphabet sizes).
Section~\ref{ternary} is devoted to ternary Gray codes and their
connections to the so-called \emph{Towers of Bucharest}.
After defining Gray codes with general radixes in Section~\ref
{general},
we extend the ternary algorithm from
Section~\ref{ternary} to arbitrary odd radixes $m$
in Section~\ref{sec:odd}, and even to mixed (odd) radixes
(Section~\ref{odd-compressed}).
In
Section~\ref{even}, we generalize the binary Gray code
algorithm of 
Section~\ref{sec:binary} to arbitrary even $m$.
Finally, in Section~\ref{sec:ahead}, we develop loopless algorithms
bases on an entirely different idea of ``working ahead'' that is related to converting
amortized running-time bounds to worst-case bounds.
In the appendix, we give prototype code for simulating all our algorithms
in \textsc{Python}.

%  The full version
%\cite{herter-arxiv} contains % all proofs as well as % prototype
%simulations of all algorithms. % in Python.

\pagebreak
\section{Ternary Gray codes and the Towers of Bucharest}
\label{ternary}
A ternary Gray code enumerates the $3^n$ $n$-tuples $(a_n,\ldots,a_1)$ with
$a_i\in\{0,1,2\}$. Successive tuples differ in one entry, and in this entry
they differ by $\pm1$.

The following simple variation of the Towers of Hanoi will yield
a ternary Gray code ($m=3$): \emph{We disallow the direct movement of
a disk between pegs $P_0$ and $P_2$}: a disk can only be moved to an
adjacent peg. We call this the Towers of Bucharest.{\footnote
{%
% % It is remarkable that the ternary code can be generated on the same
% % hardware as the binary code, just by changing the rules slightly.
%
%
% % HOW MANY DISKS? controversy.
The custom of naming variations of the Tower of Hanoi game after
different cities,
instead of using ordinary names such as ``three-in-a-row''~\cite
{three-in-a-row},
 has a long tradition.
% Toronto, Antwerpen, Amsterdam
The name ``Towers of Bucharest'' has been suggested by G\"unter
M. Ziegler. Several legends rank themselves around these towers.
\skip0=\hangindent

\leftskip=\skip0

A little count from Transylvania had conquered the whole country
% of Romania
 and
had become a powerful Lord.
In order to celebrate his glory,
he built a magnificent palace in the capital city Bucharest,
 and he suppressed his people as best he could.
He also had a dog named Heisenberg.
In a nearby monastery, the monks had golden disks of different
sizes on three pegs, and they had played the Towers of Hanoi
game for centuries. It was already forseeable that the game was drawing towards its conclusion.
According to an ancient prophecy,
the palace of the ruler of the country would crumble
and his rule would come to an end when the game would be finished.
When the count, who called himself king by this time, heard this
story,
he did not like it.
\emph
{First} of all, he had the monks beheaded and told them to do some useful
job instead.  \emph{Secondly}, he removed the pegs with the discs and took
them to his palace. He made sure that they were placed very far away from each other: The first peg was
put in the South wing, % of his palace,
the second peg in the North wing, and the third peg again in the South wing.
One may wonder why he did not place them in some more logical arrangement
like
South-South-North or South-North-North, or perhaps
 North-middle-South
or South-middle-North.
The reader will soon understand that this placement was a clever decision when
she or he learns
what else he did.
The North wing could only be reached from the South wing
through the middle wing, or by going out on the street and reentering
on the other side, but I don't think it is very wise to go into
the street carrying a heavy golden disk.

Anyway, his \emph{third} action was his
most wicked and smartest move:
%
%You will soon understand his wicked and smart decision:
%
It occurred to him that \emph{he was powerful and he was the ruler, and
 he therefore had  the power to change the rules}.
 He decreed that the discs can only be moved
between the first and the second peg or between the second and the
third peg. Direct moves between the first and third peg were henceforth forbidden.
This would delay the moves of
\parfillskip0pt
 
}}

\begin{figure}[htb]
  \centering
\noindent  \includegraphics[page=1]{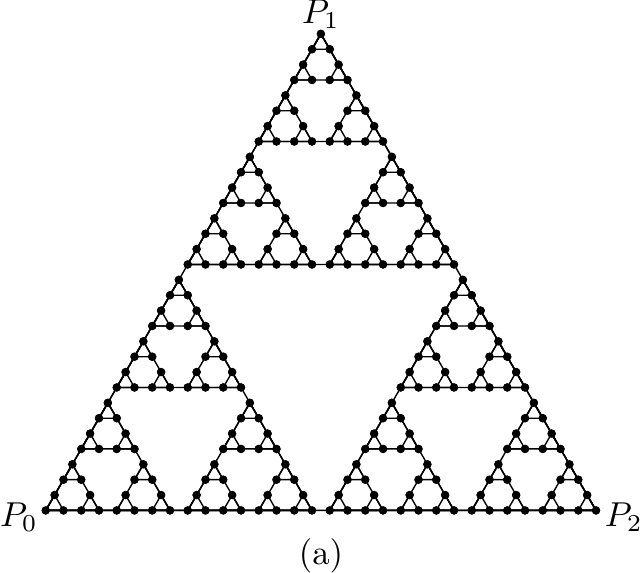}\hfill
  \includegraphics[page=2]{sierpinski}
  \caption{The state graphs of (a) the Tower of Hanoi and (b)~the Tower of
    Bucharest with $%n= %!!!!
5$ disks}
  \label{fig:state}
\end{figure}

Figure~\ref{fig:state}
 shows the state space of the Towers of
Bucharest in comparison with the Towers of Hanoi. 
In accordance with this figure,
we can make the
following easy observations:
\begin{proposition}
\label{structure}
  \begin{enumerate}
  \item 
In the towers of Hanoi,
%  at any stage, 
there are three possible moves from any position,
except when the whole tower is on one peg: In these cases,
  there are only two possible moves.
  \item 
\label{bucharest.structure}
In the towers of Bucharest,
there are two possible moves from any position,
except when the whole tower is on peg $P_0$ or $P_2$:
 In those cases,
  there is only one possible move.
  \end{enumerate}
\end{proposition}
\begin{proof}
  \begin{enumerate}
  \item 
  The disk $D_1$ can be moved to any of the other pegs (two possible
  moves). In addition, the smaller of the topmost disks on the other
  pegs (if those pegs aren't both empty) can be moved to the other peg which
  is not occupied by $D_1$.
\item 
If the disk $D_1$ is in the middle, it
 can be moved to any of the other pegs, but no other move is possible.
If the disk $D_1$ is on $P_0$ or $P_2$, it has only one possible move,
and the smaller of the topmost disks on the other
  pegs (if those pegs aren't empty) also has one possible move,
  similarly as above.
\qedhere
  \end{enumerate}
\end{proof}
Both games have the same set of $3^n$ states, corresponding to the
possible ways of
assigning each disk to one of the pegs $P_0,P_1,P_2$. % in all possible ways.
The nodes in the corners marked $P_0,P_1,P_2$ represent the states
where all disks are on one peg.
The graph of the Towers of Hanoi
in Figure~\ref{fig:state}a
 approaches the
Sierpi\'nski gasket. The optimal path of length $2^n-1$ is the
straight
path from  $P_0$ to $P_2$.
(The directions of the edges in this drawing of the state graph are not directly
related to the pegs that are involved in the exchange,
and the relation between a state and its position on the drawing is
not straightforward.)
%Figure~\ref{fig:state}
By contrast, we see that the graph of the Towers of Bucharest in
Figure~\ref{fig:state}b is a single path through all nodes.
{%\let\thefootnote\relax
\makeatletter
%\@makefntext=\long macro:
%#1->\parindent 0pt\hangindent 1em\leavevmode \hbox to 1em{\@makefnmark\hss}#1.
%   \@setfontsize\footnotesize{8.5}{9.5}%
\long\def\@makefntext#1{%
\parindent 0pt\leftskip 1em
\leavevmode %\hbox to 1em{\@makefnmark\hss}%
#1}
\def\footnotesize{\@setfontsize\footnotesize{8}{9}}\makeatother
\footnotetext{%
the disks, since they always had to be
carried all the way from the South wing to the North wing and back.
As shown in this article, the consequences of the new rule in
delaying the game are even more spectacular.

These measures were definitely overcautious,
in particular since nobody was there to move the discs any more,
and moreover, the pegs with the golden discs were carefully guarded.
Nevertheless,
he was worried that his wife and children would
wander around in the palace and play with the disks,
thereby setting the prophecy into motion again,
like in that movie, ``Jumanji'' with Robin Williams.
He was not sure how the guarding officers would behave in a conflict between
the loyalty to their orders and the authority of his family members.
You may draw your own conclusions, but in my opinion, this count, or
king if you wish, was
%certainly
pretty paranoid. In the end, it served him nothing. He was
swept away by the revolution. What became of the golden disks? Nobody
knows. It is sometimes claimed that they were hidden in a subterranean cave,
and hobby archaeologists are still looking for them occasionally.
 But probably they % disks
 have
found their way to the black market.  Today, tourists that visit the
palace are led to a stump on the floor in the North wing, which is
supposedly the remainder of one of the pegs. The South wing is
closed for restoration.

Another story, even more unbelievable but no less bloody than the
first one,
puts the Towers of Bucharest in
the context of the legendary %s around the
 caliph Harun-al-Rashid.
%incognito
One night, the caliph was again wandering through the streets of
Baghdad, as usual dressed like an ordinary businessman, in order
to assure himself that the people were still loving his reign, admiring
his wisdom, and praising his justice. He noticed
%he was drawn to 
a crowd of % on-lookers
 lookers-on who were gathered around
% who were watching 
a man and a woman sitting on the ground side by side, silently
and solemnly executing the moves of the Towers of Bucharest.  One of
them would pick up a disk and set it on an adjacent peg.  By the rules
of the game, there was always one of them for whom the two involved pegs were
easy to reach, and this was the one who carried out the move. Only on
the infrequent occasions when one of the larger and heavier disks had
to be moved, they helped each other.  The man
 wore a
%\parfillskip0pt
 
}
}

Let us see why this is true.
By Proposition~\ref{structure}, % Part~\ref{bucharest.structure},
this graph has maximum degree~2, and
 it follows that it must consist of a path between $P_0$ and $P_2$
 (the only degree-1 nodes), plus a number of disjoint cycles.
However, it is known that the path has length $3^n-1$ and does
therefore indeed go through all nodes. Since we will prove a more general
statement later
(Theorem~\ref{odd}), we only sketch the argument here:
Solving the problem recursively in an analogous way to the procedure
\textit{move\_tower}, we
reduce the problem of moving a tower of $n$ disks from $P_0$ to $P_2$
(or vice versa)
to three problem instances with
$n-1$ disks, plus two movements of disk $D_n$, and the
resulting recursion establishes that $3^n-1$ moves are required.
%to solve the problem 

The states of the Towers of Bucharest correspond in a natural way to
the ternary $n$-tuples: The digit $a_i\in\{0,1,2\}$ gives the position
of disk $D_i$.  It follows now easily that the solution of the Towers of
Bucharest
%goes through all te
yields a ternary Gray code: Since we can move only one disk at a time,
it means that we change only one digit at a time, and by the special
rules of the Towers of Bucharest, we change it by $\pm1$.
(This connection has apparently not been made before.)
In fact, the algorithm produces \emph{the} ternary reflected Gray
code,
which we are about to define
%as defined 
below 
in Section~\ref{general}; see also
Theorem~\ref{odd}. 

Moreover, since there are only two possible moves, one just has to
always choose the move which does not undo the previous move, and this
leads to a very easy loopless Gray code enumeration algorithm.

It is remarkable that ternary Gray codes
 can be
generated on the same hardware as binary Gray codes~(Fig.~\ref{hanoi}).
In the context of %When
 generating the ternary Gray code, the %example
Gray code string can be directly read off the disks.
For example,
the configuration in
Fig.~\ref{hanoi}
represents the string
211102. % P0:2 P1:5,4,3 P2:6,1
It is $D_1$'s turn to move, and
%From there, 
the disk $D_1$ will make two steps to the left, generating the strings
211101 % P0:2 P1:5,4,3,1 P2:6
and
211100, % P0:2,1 P1:5,4,3 P2:6
and pauses there for one step, while
 disk $D_3$ moves to the right, leading to the string
211200, % P0:2,1 P1:5,4 P2:6,3
%211201 P0:2 P1:5,4,1 P2:6,3
etc.

\section{Gray codes with general radixes}
\label{general}
An $m$-ary Gray code enumerates the $n$-tuples $(a_n,\ldots,a_1)$ with
$0\le a_i<m$, changing a single digit at a time by $\pm1$. 
The $m$-ary reflected Gray code can be recursively described as
follows:
Let $C_1,C_2,\ldots,C_{m^{n}}$ be the Gray code for the strings of
length $n$. Then the strings of length $n+1$ are generated in the
order
\begin{equation}
\label{m-ary}
  \begin{array}[c]{l}
  C_10,C_11,C_12,\ldots , C_1(m-2),C_1(m-1),\ \
  C_2(m-1),C_2(m-2),\ldots , C_22,C_21,C_20,\\
  C_30,C_31,C_32,\ldots , C_3(m-2),C_3(m-1),\ \
  C_4(m-1),C_4(m-2),\ldots , C_42,C_41,C_40,\\
  C_50,C_51,C_52,\ldots , C_5(m-2),C_5(m-1),\ \
  \ldots
  \end{array}
\end{equation}
Each digit alternates between an upward sweep from $0$ to $m-1$
and a return sweep from $m-1$ to $0$.
{%\let\thefootnote\relax
\makeatletter
\long\def\@makefntext#1{%
\parindent 0pt \leftskip 1em
\leavevmode %\hbox to 1em{\@makefnmark\hss}%
#1}
%   \@setfontsize\footnotesize{8.5}{9.5}%
\def\footnotesize{\@setfontsize\footnotesize{6.5}{7.4}}\makeatother
\footnotetext{%
\tiny
 cowboy hat,
 and the woman was in her bikini.  After all, it might have been the Towers
of Hanoi that they played. Some witnesses have later reported that they had
seen a disc jumping between the first and the third peg, but this has
never been conclusively confirmed.

May that as it be, something unexpected happened.  As the khaliph was
engrossed in watching the spectacle %two players
 and drew %moved
 a bit closer, a
small door in the wall beside him opened, which he had not noticed
before.  It %He was 
gave %led
%opened
 onto a small garden.
The moon
had risen over the rooftops, and her light gave a sort-of surreal
atmosphere to the whole scene. 
In the middle of
the garden,
at the corner of a fountain,
 a woman sang, accompanying herself on the lute.  She had a
beautiful voice, a bit like Mariah Carey or Adele.  
 The calif would have listened longer,
%to the chant,
 but he was quickly escorted into a house, where a maid-servant took
charge of him and handed him a black gown.  ``Hurry up, you are late.
We were waiting only for you!''  The gown covered his whole stature
and hid his face, and he entered a room that was barely lit by an open
fire. Seven other men % persons
 in %similar
 black gowns were already sitting on small
stools in a circle around the fire. One stool was free, and he
sat down.
Beside the fire, there was a small ivory model of the Towers of Bucharest, with the four largest of the $n=6$ disks
 on the final peg. Disc~1 was on the middle peg, and disk~2 was on the
 first peg.  The kaliph, having watched the game just before, understood immediately what that meant.
 Nowbody said a word. The tension rose.
After six minutes, a lady entered and addressed them.  She was the singer from
the fountain.  ``Gentlemen. You have sworn to come to my rescue when
I would be in need.  Now the time has come to fulfill your oath.
%
%This peg contains 
You see seven discs of different sizes.
He %The one 
who will draw the smallest disk
%
%  This
% necklace of mine contains nine white pearls and one black pearl.  The
% one who will draw the black pearl 
%
will bring me the head of the detestable % odious %hated
 caliph
Harun-al-Rashid
(ca.~763--809).  Should he fail to fulfill this task, the other
eight will kill \emph{him}, and we will come together and draw again.''  With these
words,
she dropped the discs into
% she tore her necklace and dropped the pearls into
%an urn
a chalice.
%It is undeniable that 
%Now, the tension had certainly % had now risen
%increased even more in the room.
 In silence, each man picked a
% pearl.
disk.
 The kaliph was last to draw.
As he opened his hand, sure enough, he found the
smallest disc, disc number~1.
% black pearl.
He rose and said:
``Fair lady, I will fulfill your order as I have promised.
But pray tell me:
by which deeds or words has the kalif 
enraged you so much that
you wish him to%  be%
\break
% killed?''
 
}
%\addtocounter{footnote}{-2}
}

\section{Generating the $m$-ary Gray code with odd $m$}
\label{sec:odd}

For odd $m$, the ternary algorithm can be generalized.
We need $m$ pegs $P_0,\ldots,P_{m-1}$. The leftmost peg $P_0$ and the
rightmost peg $P_{m-1}$ play a special role. The other pegs are called
the \emph{intermediate pegs}.
\begin{tabbing}
    \quad\=\+
\emph{Algorithm} \textbf{ODD}. Generation of the $m$-ary Gray code for odd $m$.\\
    \quad\=\+
Initialize:
Put all disks on $P_0$.\\
    \textbf{loop}:\\
    \quad\=\+
    Move $D_1$ for $m-1$ steps, from $P_0$ to $P_{m-1}$ or vice versa.
\\
Let $D_k$ be the smallest of the topmost disks on the $m-1$ pegs that
don't carry $D_1$.
\\
If there is no such disk, terminate.
\\ Move $D_k$ by one step:\\
    \quad\=\+
If $D_k$ is on $P_0$ or $P_{m-1}$, there is only one possible direction where
to go.\\
Otherwise, the disk $D_k$ continues in the same direction as in its last move.
\end{tabbing}

Figure~\ref{fig:klagenfurt} shows an example with $m=5$. The game with
5 pegs is called the Tower of Klagenfurt, after the birthplace of the
senior author.\footnote{%
When %the city of
 Klagenfurt was founded, it was surrounded by a swamp.
The swamp was inhabited by a dinosaur, the so-called \emph{Lindworm}.
The Lindworm would regularly come to the city and eat %some
citizens. Occasionally, she would devour one of the towers of the city.
The coat of arms of Klagenfurt shows the Lindworm dragon 
in front of
%with
 the only
remaining tower,
see
Figure~\ref{fig:klagenfurt}.
 (Initially, there were five towers.) Over the centuries,
the swamp has been drained, and the
Lindworm is practically extinct.
%\parfillskip=0pt
}

\begin{figure}[htb]
  \centering
$\vcenter{\hbox
{\includegraphics[width=0.67\textwidth]{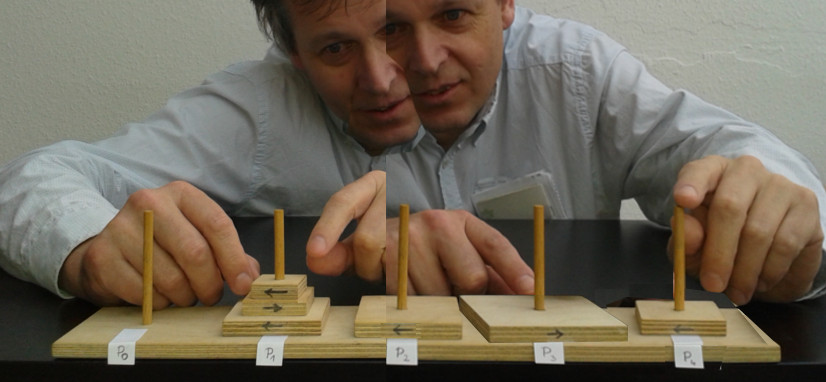}}}$
\hskip 8mm
$\vcenter{\hbox
{\includegraphics[scale=0.27]{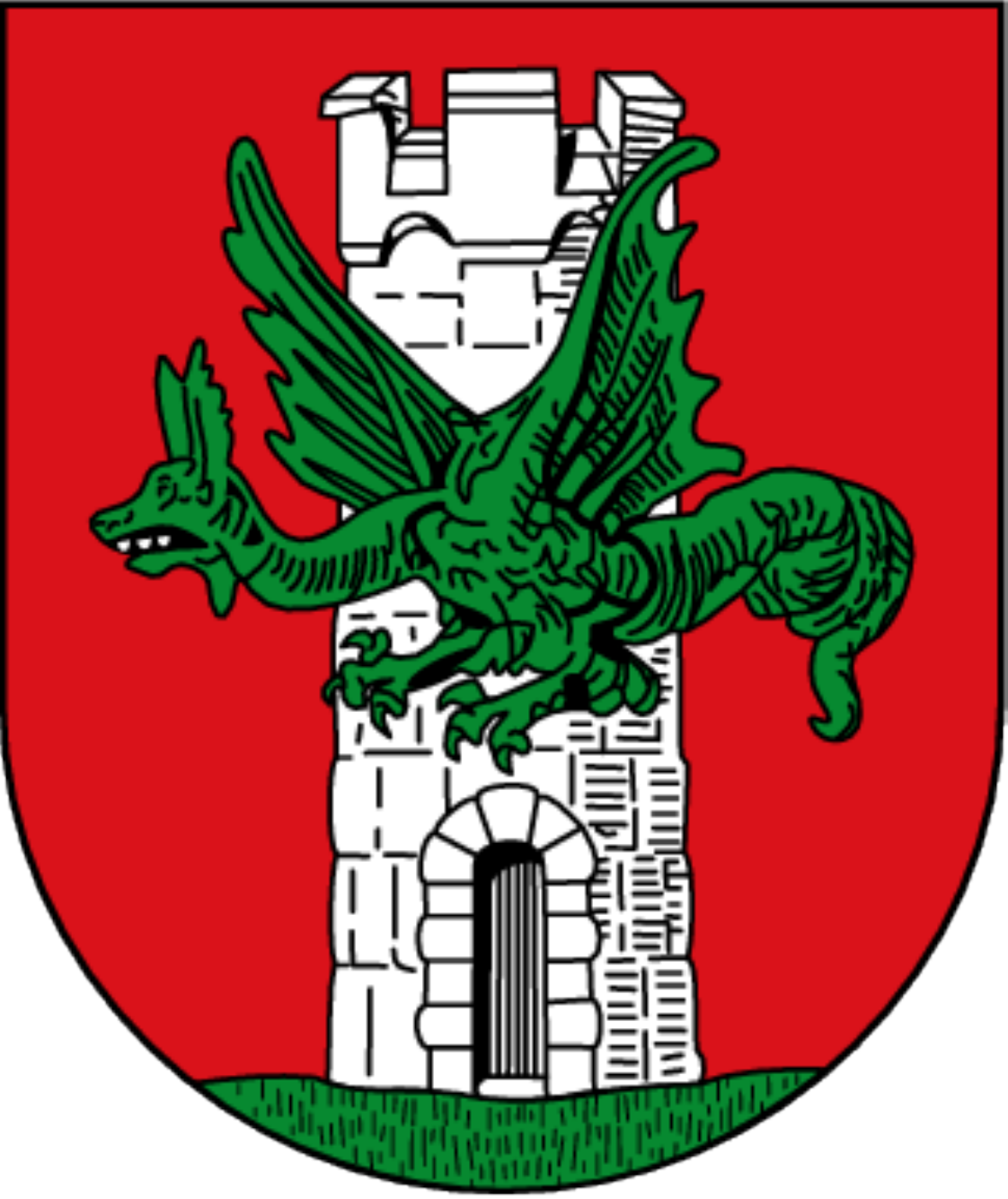}}}$

% 321411 P0: P1:(6, 3, 1),(5, 2, -1),(4, 1, -1),(2, 1, 1),(1, 1, -1) P2:(3, 4, -1)

% 321410 P0:(1, 0, 1) P1:(6, 3, 1),(5, 2, -1),(4, 1, -1),(2, 1, 1) P2:(3, 4, -1)

% 321420 P0:(1, 0, 1) P1:(6, 3, 1),(5, 2, -1),(4, 1, -1),(2, 2, 1) P2:(3, 4, -1)

% 211030 P0: P1:4,2,1 P2:5 P3:6 P4:3
  
  \caption{The Towers of Klagenfurt.
This configuration represents the string
321411
over the radix % alphabet with
$m=5$. % letters. 
The next step of the Gray code moves the smallest disk $D_1$ onto peg $P_0$, changing the string to
321410.
% P0:(1, 0, 1) P1:(6, 3, 1),(5, 2, -1),(4, 1, -1),(2, 1, 1) P2:(3, 4, -1)
After that, disk $D_2$  moves from $P_1$ to $P_2$ and the next string is
321420. % P0:(1, 0, 1) P1:(6, 3, 1),(5, 2, -1),(4, 1, -1),(2, 2, 1)
        % P2:(3, 4, -1)
 In the background, the two-headed Lindworm monster.
}
  \label{fig:klagenfurt}
\end{figure}

In this procedure, the movement of $D_1$ is ``externally given'',
whereas the movement of the other disks, whenever $D_1$ is at rest, is somehow
``determined by the algorithm''.
It is not obvious that the algorithm does not put a larger disk on top
of $D_1$.
\begin{theorem}
\label{odd}
  Algorithm ODD generates the $m$-ary reflected Gray code defined
in~\eqref{m-ary}.
\end{theorem}
\begin{proof}
  It is clear from the algorithm that the last digit, which is controlled by the movement
  of $D_1$, changes in accordance with~\eqref{m-ary}.
We still have to show that when we discard the last digit and observe only
the movement of the disks $D_2,\ldots,D_n$, the algorithm produces
the Gray code for the strings of length $n-1$. This is proved by
induction.

By the rules of the algorithm, whenever $D_1$ rests, the disk that
moves is $D_2$, unless $D_2$ is covered by $D_1$.  Let us now observe
the motion pattern of $D_1$ and $D_2$ that results from this rule.  We
start with $D_1$ on top of $D_2$, say, on peg $P_0$, with $D_1$ about
to start its sweep. Whenever $D_1$ pauses for one step,
 $D_2$ will
make a step towards $P_{m-1}$. After $D_2$ reaches
 $P_{m-1}$,
 it turns out that, because $m$ is odd,
$D_1$ will make its next sweep from
$P_0$ to $P_{m-1}$, resting on top of $D_2$. Now, since $D_2$ is
covered, it will be one of the \emph{other} disks $D_3,D_4,\ldots$
that will move. Then the same routine repeats in the other direction.

If we now ignore $D_1$ and look only at the motions of the other
disks, the following pattern emerges:
$D_2$ makes $m-1$ steps from one end to the other, and then
the smallest disk that is not covered by $D_2$ makes its move,
according to the rules.
This is precisely the same procedure as Algorithm ODD, with $D_2$
taking the role of the ``externally controlled'' disk $D_1$,
and we have assumed by induction that this algorithm correctly produces the
 Gray code for the strings of length $n-1$, and it does not put a
 larger disk on top of $D_2$. Since the larger disks are moved only
 when
$D_2$ lies under $D_1$, it follows that a larger disk cannot be moved
on top of $D_1$ either.
\end{proof}

One can actually apply one induction step of the proof in the opposite direction, introducing
an additional ``control disk'' $D_0$ which does not have a digit
associated with it. Its only role is to alternately cover $P_0$ and
$P_{m-1}$ and exclude these pegs from the selection of the disk $D_k$
that should be moved.
The algorithm becomes simpler because it does not have to treat $D_1$
separately from the other disks. (See
Appendix~\ref{imple-odd-compressed},
where this idea is applied to the algorithm of
Section~\ref{odd-compressed} below).

\section{Generating the $m$-ary Gray code with even $m$}
\label{even}

For even $m$, we generalize 
Algorithm HANOI, which solves the case $m=2$.
We use $m+1$ pegs $P_0,\ldots,P_{m}$, which we arrange in a cyclic
clockwise order.
 We stipulate that
disks $D_i$ with odd $i$ move only clockwise, and 
disks with even $i$ move only counterclockwise.
\begin{tabbing}
    \quad\=\+
\emph{Algorithm} \textbf{EVEN}. Generation of the $m$-ary Gray code for even $m$.\\
    \quad\=\+
Initialize:
Put all disks on $P_0$.\\
    \textbf{loop}:\\
    \qquad\=\+
    Move $D_1$ for $m-1$ steps, in clockwise direction.
\\
Let $D_k$ be the smallest of the topmost disks on the $m$ pegs that
don't carry $D_1$.
\\
If there is no such disk, terminate.
\\ Move $D_k$ by one step, in the direction determined by the parity
of $k$.
\end{tabbing}
The Gray code is determined by changing the digit $a_i$ whenever disk
$D_i$ is moved. The digit $a_i$ runs through the sequence
$0,1,2,\ldots,m-2,m-1,m-2,\ldots,2,1,0,1,2,\ldots$. Thus it changes always
by $\pm1$, but we have to remember whether it is on the increasing or
the decreasing part of the cycle.
The position of %direction of movement of
disk $D_i$
 is no longer directly correlated
with the % change of the
 digit $a_i$;
%Moreover, there is no direct correspondence between the position and
 thus the digits $a_i$ have to be maintained
separately, in addition to the disks on the pegs. It is far
from straightforward to relate the disk configuration to the Gray code.

For example, when carrying out the algorithm
for $m=4$, the configuration in Figure~\ref{fig:klagenfurt} appears when
the string is
211030. %P0: P1:4,2,1 P2:5 P3:6 P4:3
Disk $D_1$ has just made three steps and is going to rest for one step.
The next step moves $D_3$ clockwise, since $3$ is odd, and the string
is changed to
211130. % P0:3 P1:4,2,1 P2:5 P3:6 P4:
After that, $D_1$ resumes its clockwise motion, and the string changes into
211131. % P0:3 P1:4,2 P2:5,1 P3:6 P4:

\begin{theorem}
  Algorithm EVEN generates the $m$-ary reflected Gray code defined
in~\eqref{m-ary}.
\end{theorem}
\begin{proof}
This follows along the same lines as Theorem~\ref{odd}.
When we look at the pattern of motion of $D_1$ and $D_2$, we
 observe again that $D_2$ makes $m-1$ steps until it is covered by
 $D_1$,
see Fig.~\ref{fig:pentagon}:
%  After the first sweep of $D_1$, the
% counterclockwise
% cyclic
% distance from $D_1$ to $D_2$ is $m-2$
 After the first move of $D_2$, the
clockwise
cyclic
distance from $D_1$ to $D_2$ is $1$, and with each move of $D_2$, this
distance increases by~1. Thus, after $m-1$ moves, the distance becomes
$m-1$, and $D_1$ will land on top of $D_2$ with its next sweep.
\end{proof}
\begin{figure}[htb]
  \centering
  \includegraphics[scale=0.8]{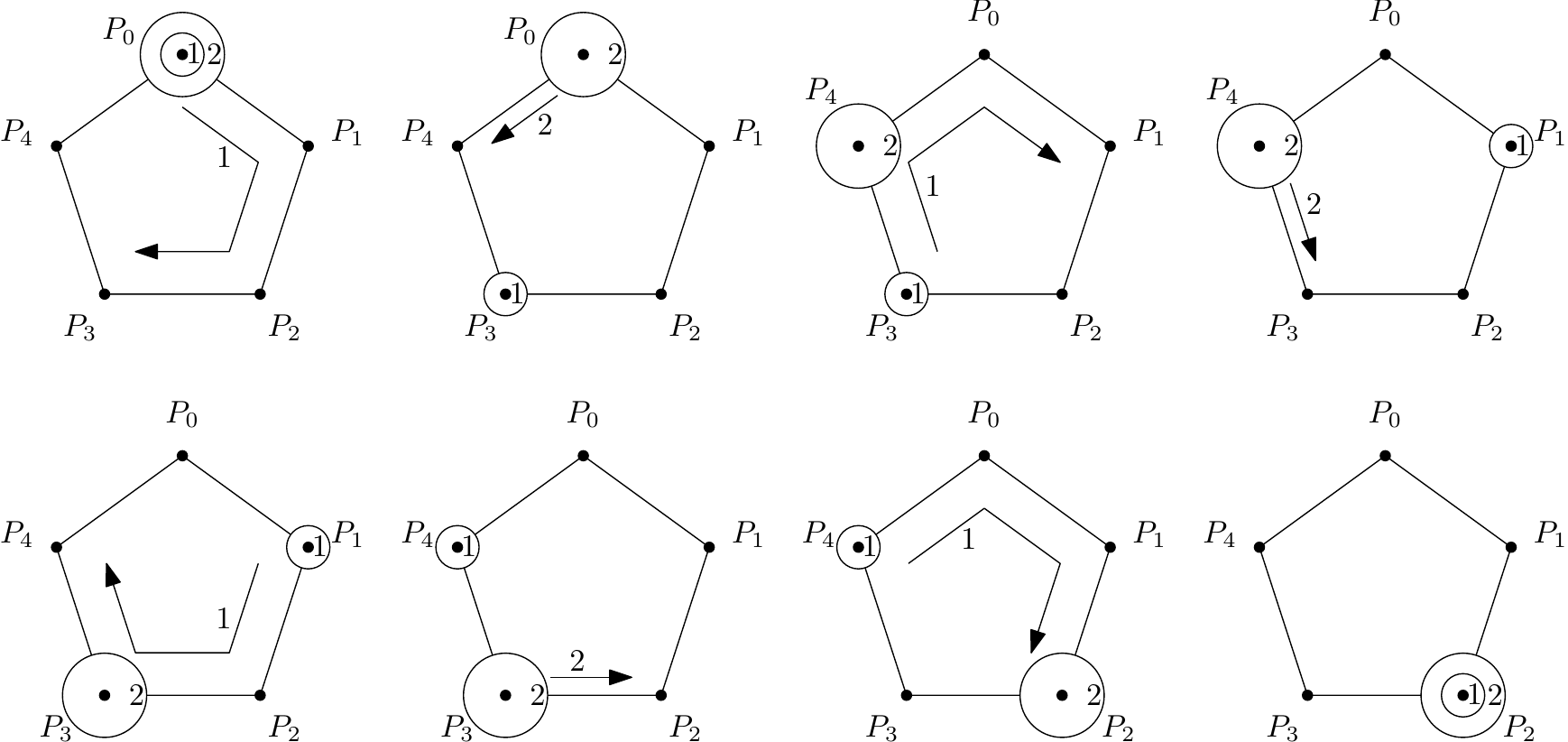}
  \caption{One period of movement of the two smallest disks $D_1$ and
    $D_2$
when Algorithm EVEN generates all tuples over an alphabet of size $m=4$ using $m+1=5$ pegs.}
  \label{fig:pentagon}
\end{figure}

 Algorithms ODD and EVEN
 do not generate a shortest sequence of moves to the
target configuration except when $m=3$ or $m=2$.
We could not come up with some set of natural constraints under which
our algorithms
give a shortest solution.

\section{The Towers of Bucharest++}
\label{odd-compressed}

In Algorithm ODD, the intermediate pegs $P_1,\ldots,P_{m-2}$ will
always be available for selecting the smallest disk $D_k$ to be moved.
Thus, one can coalesce these pegs into one peg, keeping only
 the two extreme pegs separate.
With three pegs, we can use the same hardware as the tower of Bucharest,
but we have to record the value of the digits, since they are no
longer expressed by the position. A simple method is to provide the disks
with \emph{marks} that indicate the value as well as the direction of
movement, which we have to remember anyway.
Each disk cycles through $2m-2$ values, potentially augmented with
direction information:
\begin{equation}
  \label{eq:up-down}
0,\ 
 1\!\uparrow,\ 
 2\!\uparrow,\ 
\!\ldots,\ 
 (m-2)\!\uparrow,\ 
m-1,\ 
 (m-2)\!\downarrow,\ 
\!\ldots
 2\!\downarrow,\ 
 1\!\downarrow,\ 
0,\ 
 1\!\uparrow,\ 
\ldots
\end{equation}
 It makes sense to encode this information like a dial with $2m-2$
equally spaced directions, as shown in Fig.~\ref{fig:dials}a.
A disk whose mark shows 0 is always on the left peg $P_0$.
A~disk whose mark shows $m-1$ is always on the right peg~$P_2$.
Otherwise, it is on the middle peg~$P_1$.
When we say we \emph{turn a disk}, this means
that we turn it clockwise to the next dial position,
and if necessary, move it to the appropriate peg.
\begin{figure}[htb]
  \centering
\tabskip=0pt plus 1fill
\halign to \hsize{&\hfil#\hfil\cr
  \includegraphics[scale=0.8]{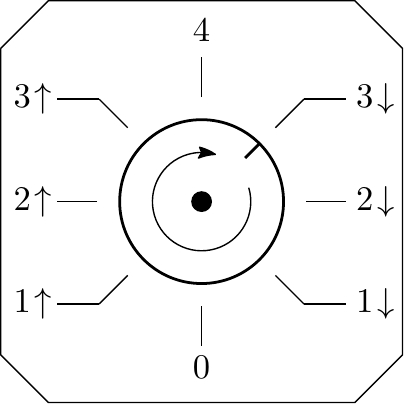}
&  \includegraphics{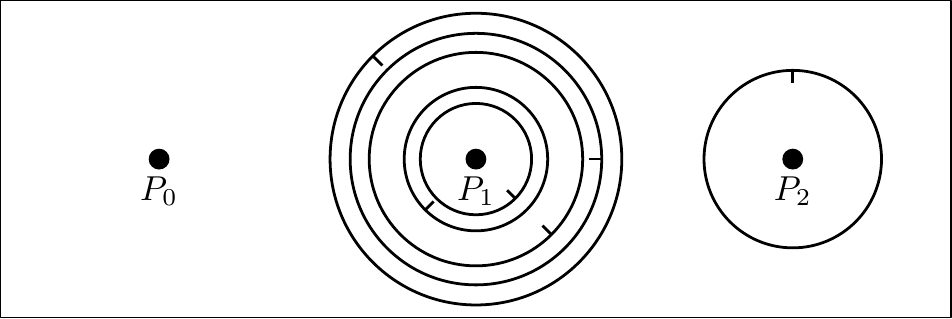}
\cr
(a)&(b)\cr
}
  \caption{(a) The upgraded disk of the Towers of Bucharest++ and the
    meaning of its positions, for $m=5$.
(b) The situation of Figure~\ref{fig:klagenfurt}, compressed to 3 pegs.
}
  \label{fig:dials}
\end{figure}
%
%
%  the following operation.
% \begin{quote}
%   (\emph{Turning a disk $D$})
% Turn $D$ clockwise to the next dial position, and move it to the corres
% \end{quote}
%
\begin{tabbing}
    \quad\=\+
\emph{Algorithm} \textbf{ODD-COMPRESSED}. Generation of the $m$-ary Gray code for odd $m$.\\
    \quad\=\+
Initialize:
Put all disks on $P_0$, and turn them to show~0.\\
    \textbf{loop}:\\
    \qquad\=\+
Turn disk $D_1$ $m-1$ times until it arrives at one of the extreme
pegs $P_0$ or $P_2$.
\\
Let $D_k$ be the smaller of the topmost disks on the two pegs not
covered by $D_1$.
\\
If there is no such disk, terminate.
\\ Turn $D_k$ once.
\end{tabbing}
The digits $a_i$ can be read off from the dial positions.
Correctness follows by comparison with Algorithm ODD,
checking that the transition between successive states is
preserved when merging the intermediate pegs into one peg.
\qed

This algorithm can now even be generalized to mixed-radix Gray codes for 
 the $n$-tuples $(a_n,\ldots,a_1)$ with
$0\le a_i<m_i$, for some sequence of radixes $m_i\ge 2$, provided
that all $m_i$ are odd.

% {\let\thefootnote\relax
% \footnote{B123}
% }

\section{Simulation}
\label{simulation}
All our algorithms can be easily simulated in software on a digital
computer.\footnote
{Nowadays, most households will more readily have access to a computer than
to towers of Hanoi.}
A stack will do for each peg.
If there are $k$ pegs, the algorithm takes $O(k)$ time to
compute the next move and accordingly produce the next element of the
Gray code sequence. If $m$ is constant, then $k=m$ or $k=m+1$ in
Algorithms ODD and EVEN, and these algorithms can pass as loopless
algorithms.
If $k$ is large, Algorithm ODD can be replaced by ODD-COMPRESSED,
which has only 3 pegs, independent of $m$.

%Except for the algorithm with even $m$ (Section~\ref{even}), it takes
%constant time to
To make a truly loopless algorithm out of Algorithm EVEN,
at the expense of an
increased overhead,
 we can use
the following easy fact, which follows directly from the algorithm statement.
\begin{lemma}
In the algorithms EVEN, ODD, and ODD-COMPRESSED,
  when a disk $D_k$ is moved, all smaller disks $D_1,\ldots,D_{k-1}$ are
  on the same peg.
\qed
\end{lemma}
%This opens the way for
To get a loopless implementation,
% of the algorithms in Section x and y: 
the set of disks on a peg has to be maintained as a
sequence of maximal intervals of successive integers, instead of
storing them as a stack in the usual way. Then, whenever $D_1$ is at
rest, the disk $D_k$ to be moved can be determined in constant time as the smallest
missing disk on the peg containing $D_1$. 
%involved andmore complicated.

%\input workingahead
\section{Working ahead}
\label{sec:ahead}

While we are at the topic of Gray codes, we might as well mention
another approach for loopless generation of Gray codes, which results
from a generally applicable technique for converting amortized bounds
into worst-case bounds.
We start from the observation that was already mentioned in connection
with the delta-sequence in Section~\ref{delta}:
\begin{proposition}
Consider the enumeration of
 the $n$-tuples $(b_n,\ldots,b_1)$ with
$0\le b_i<m_i$ in \emph{lexicographic order}.
If, between two successive tuples of
the sequence, 
the $j$ rightmost digits are changed,
then,
at the corresponding transition in the Gray code,
 the $j$-th
digit from the right is changed.
\qed
\end{proposition}
We can thus find the position $j$ that has to be changed in the Gray code by lexicographically
``incrementing'' $n$-tuples $(b_n,\ldots,b_1)$ in a straightforward way:

\pagebreak

\begin{tabbing}
    \quad\=\+
\emph{Algorithm} \textbf{DELTA}. Generation of the delta-sequence for the
Gray code.\\
    \quad\=\+
Initialize:
$(b_{n},\ldots ,b_2,b_1) := (0,0,\ldots ,0,0)$\\
 $Q$ :=  an empty list\\
    \textbf{loop}:\\
    \qquad\=\+
$j := 1$\\
\textbf{while} $b_j=m_j-1$:\\
    \qquad\=\+
  $b_j := 0$\\
  $j := j+1$\\
\textbf{if} $j=n+1$: TERMINATE
\-\\
$b_j := b_j+1$\\
$Q.\textit{append}(j)$
\end{tabbing}
%The algorithm stores 

The delta sequence is stored in $Q$.
It is known that the \emph{average} number
of loop iterations
% (``steps'') 
for producing an entry of $Q$ is less than~2.
We use this fact to coordinate the \emph{production} of entries
$Q$ by Algorithm DELTA with their \emph{consumption} in the Gray code
generation, turning $Q$ into a buffer of bounded capacity.
This leads to the following loopless algorithm:

%\\
\bigskip

\noindent
\quad
 \begin{minipage}[t]{1\linewidth}
\begin{tabbing}
 \quad\=  \quad\=  \qquad\=  \qquad\=  \qquad\=  \qquad\=
 \qquad\=\kill
\emph{Algorithm} \textbf{WORK-AHEAD}.\\ Generation of the Gray code.\\
\textbf{procedure}
STEP:\+\\
\textbf{if} $b_j=m_j-1$:\+\\
  $b_j := 0$\\
  $j := j+1$\-\\
\textbf{else}:\+\\
        \textbf{if} $Q$ is not filled to capacity:\+\\
 $b_j := b_j+1$\\
$Q.\textit{append}(j)$\\
$ j := 1$
\end{tabbing}
 \end{minipage}
\quad\vrule\quad
\begin{minipage}[t]{0.5\linewidth}
\begin{tabbing}
 \qquad\=  \qquad\=  \qquad\=  \qquad\=  \qquad\=  \qquad\= \qquad\=\kill
{$(a_{n},\ldots, a_2,a_1) := (0,\ldots ,0,0)$}\\
$(d_{n},\ldots, d_2,d_1) := (1,\ldots ,1,1)$\\
$(b_{n+1},b_{n},\ldots, b_2,b_1) := (0,0,\ldots ,0,0)$;
$m_{n+1} := 2$\\
  $Q$ := queue of capacity $B := \lceil \frac {n}2\rceil$, initially empty\\
  $j:=1$\\
  \textbf{loop}:%\\  Main loop:
\+\\
{visit the $n$-tuple $(a_{n},\ldots ,a_2,a_1)$}\\
      STEP\\
      STEP\\
     remove $j$ from $Q$\\
\textbf{if} $j=n+1$: TERMINATE\\
$a_{j} := a_{j} +d_j$\\
\textbf{if} $a_j=0$ or $a_j=m_j-1$: $d_j := -d_j$
\end{tabbing}
\end{minipage}
\bigskip

The procedure STEP on the left side encompasses one loop iteration of
Algorithm DELTA.
By programmer's license, we have moved
% We have taken the programmer's license to move
 the initialization
$ j := 1$ of the loop variable to the end of the previous loop. We have also moved the
termination test $j=n+1$ to the side of the consumer.
Accordingly, we had to extend the $n$-tuple $b$ into an $(n+1)$-tuple,
setting $m_{n+1}$ arbitrarily to 2.
When $Q$ is full, nothing is done in the procedure STEP, and the repeated
call of STEP will try to insert the same value into $Q$. Thus, apart
from the termination test, a repeated execution of STEP will
faithfully carry out Algorithm DELTA.

The Gray code algorithm on the right couples two production STEPs with
one consumption step, which takes out an
entry $j$ of $Q$ and carries out the update $a_j := a_j \pm 1$.
Every digit $a_j$ must cycle up and down through its values in the
sequence \eqref{eq:up-down}, and thus, we have to remember the
direction $d_j=\pm1$ in which it moves, as in Algorithm~ODD.
%, in oder to determine whether
%$a_j$ is being incremented or decremented.

To show that the algorithm is correct, we have to ensure
that the % two things:
%\begin{enumerate}
%\item The 
queue $Q$ is never empty when the algorithm retrieves an
  element from it.
This is proved below in
Lemma~\ref{queue-not-empty}.
% in Appendix~\ref{work-ahead-proof}.

%\item
 The clean way to terminate the algorithm would be to stop inserting
elements into $Q$ as soon as $j=n+1$ is \emph{produced} in STEP,
as in Algorithm DELTA.
Instead, termination is triggered when the value
 $j=n+1$ is \emph{removed} from
  $Q$.
 Due to this delayed
termination test, a few more iterations of STEP can be carried out,
but they cause no harm.
%
% We have to show that there is no harm in potentially
%executing a few more iterations of
%\end{enumerate}

For the \emph{binary} Gray code ($m_i=2$ for all $i=1,\ldots,n$),
the algorithm can be simplified. With a slightly larger
buffer $Q$ of size
% $B' := \lceil \frac {n+1}2\rceil$,
$B' := \max\{\lceil \frac {n+1}2\rceil,2\}$,
 the test
whether
 $Q$ is filled to capacity can be omitted,
see Lemma~\ref{binary} below.
% in Appendix~\ref{work-ahead-proof}. 
The reason
is that the average number
of production STEPs per item approaches 2 in the limit, and
accordingly, the queue automatically does not grow beyond the minimum
necessary size.
The directions $d_i$ are of course also superfluous in the binary case.

%%%%%%%%%%%%%%%%%%%%%%%%%%%
The idea of ``working ahead''
%in Section~\ref{sec:ahead}
 is opposite to the approach of delaying
work as long as possible
that underlies many
``lazy'' data structures and also
lazy evaluation in some functional programming languages.
In a similar vein,
Guibas,  McCreight,  Plass, and Janet R. Roberts~\cite{Guibas:1977}
have obtained worst-case bounds of $O(\log k)$ for updating a sorted linear
list at distance $k$ from the
beginning.\footnote
{We thank Don Knuth for leading us to this reference.}
Their algorithm works ahead to hedge against sudden bursts of activity.
Our setting is much simpler, because we do not depend on the update
requests of a ``user'' and we can plan everything in advance.

 At a different level of complexity, the idea of working ahead occurs in an algorithm of
{Wettstein}~\cite[Section~6]{Wettstein}.
% {Counting and enumerating crossing-free geometric graphs}
% ETH Z\"urich,
% 2014
This trick, credited to Emo Welzl, is used to achieve \emph{polynomial delay} between
 successive solutions when
 enumerating non-crossing perfect matching of a planar point set,
despite having to build up a network with exponential space in a
preprocessing phase.

\subsection{An alternative STEP procedure} % for WORK-AHEAD (Section \ref{sec:ahead})}
\label{sec:alternative-STEP}

As an alternative to the organization of Algorithm WORK-AHEAD,
% in Section~\ref{sec:ahead}, 
we can incorporate the termination test into
the STEP procedure:

\begin{tabbing}
 \quad\=  \quad\=  \qquad\=  \qquad\=  \qquad\=  \qquad\=
 \qquad\=\+\kill
\textbf{procedure}
STEP$'$:\+\\
\textbf{if} $j=n+1$: TERMINATE\\
\textbf{if} $b_j=m_j-1$:\+\\
  $b_j := 0$\\
  $j := j+1$\-\\
\textbf{else}:\+\\
        \textbf{if} $Q$ is not filled to capacity:\+\\
 $b_j := b_j+1$\\
$Q.\textit{append}(j)$\\
$ j := 1$
\end{tabbing}

With this modified procedure STEP$'$, the termination test in
the main part of Algorithm WORK-AHEAD can of course be omitted.
We also need not extend the arrays $b$ and $m$ to $n+1$ elements.

 The algorithm still works correctly because there
are no unused entries in the queue when STEP$'$ signals termination.
Let us prove this:

The termination signal is sent instead of producing the value $j=\bar\rho(k)=n+1$
for $k=m_0m_1\ldots m_{n-1}$. Generating this signal takes
 $n+1$
iterations of STEP. In this time, no new values are added to the
queue.
Let us assume % the extreme scenario 
that
 the production of 
$\bar\rho(k)%=n+1
$ was started during iteration $k_0$,
% in the first repetition of STEP, 
and the buffer was filled
%to the maximum possible extent of
 with $B_0\le B$ entries at that time.
The first of these entries is consumed at the end of iteration $k_0$, and
all $B_0$ entries of the
buffer
have been used up at the beginning of iteration $k_0+B_0$.
By this time,
at most $2B_0\le 2B\le n+1$ iterations of STEP were carried out
and contributed to the production of the termination signal.
It follows that when STEP discovers that $j=n+1$, no unused
entries are in the stack, and it is safe to terminate the program.

It is important not to ``speed up'' the program by moving the
termination test into the \textbf{if}-branch after the statement
$j:=j+1$.
Also, we must use exactly the prescribed
buffer size for~$Q$.
Therefore, this variation is incompatible with the simplification for the binary case
mentioned above. % at the end of Section~\ref{sec:ahead}.

\subsection{Correctness proofs for the work-ahead algorithms}
%Section \ref{sec:ahead}}
%\label{work-ahead-proof}

%Let us first analyze the running time for each iterations of
% Algorithm DELTA. 
We define the ruler function $\rho$
and the modified ruler function $\bar\rho$
with respect to
a sequence of radixes $m_1,\ldots,m_n$ as follows:
\begin{equation*}
 \rho(k) := \max \{\, i : 0\le i\le n,\ m_1m_2\ldots m_i\text{ divides
   }k\,\},\qquad
\bar
 \rho(k) := 
 \rho(k)+1
\end{equation*}
Then the $k$-th value that is entered into $Q$ is $\bar \rho(k)$, and
for computing this value,
 Algorithm DELTA needs $\bar\rho(k)$  iterations,
and accordingly, Algorithm WORK-AHEAD needs $\bar\rho(k)$ STEPs.

\begin{lemma}\label{queue-not-empty}
  In Algorithm WORK-AHEAD, the buffer $Q$ never becomes empty.
\end{lemma}
\begin{proof}
  We number the iterations of the main loop as
  $1,2,\ldots,m_1m_2\ldots m_n$. In the last iteration,
the algorithm terminates.

Let us show that the queue $Q$ is not empty in iteration $k$.
We distinguish two cases.
\begin{enumerate}
\item [(i)]
Up to and including iteration $k$, two repetitions of STEP were always completed.
\item [(ii)]
Some repetitions of STEP had no effect because the buffer $Q$ was
full.
\end{enumerate}
In case~(i), production of all values $\rho(i)$ for $i=1,\ldots,k$
 requires
 \begin{equation*}
  S(k) := \sum_{i=1}^k \bar\rho(i) 
 \end{equation*}
calls to STEP. To show that these calls are completed by the time when $\bar\rho(k)$ is
needed, we have to show
\begin{equation}\label{left}
  S(k)\le 2k. 
\end{equation}

In case~(ii), let $k_0$ be the last iteration when an execution
of STEP was ``skipped''.
This means that the queue $Q$ was filled
to capacity $B$ just before removing the value $j=\bar\rho(k_0)$,
and it
%At this time, the queue $Q$
 contained the values
$\bar\rho(k_0),\bar\rho(k_0+1),\ldots,\bar\rho(k_0+B-1)$. Since then, STEP was called $2(k-k_0)$ times,
and
$\bar\rho(k)$ is ready when it is needed, provided that
\[1 +\sum_{i=k_0+B+1}^k \bar\rho(i)\le 2(k-k_0)\]
whenever $k\ge k_0+B$.
The left-hand side of this inequality is the number of necessary STEPs
for computing the values up to $\bar\rho(k)$.
Computing $\bar\rho(k_0+B)$ takes just one more STEP, since the STEP
that would have stored this value in $Q$ was abandoned in iteration $k_0$.
Setting $k'=k_0+B$, we can express the inequality equivalently as
\begin{equation}\label{interval}
S(k)-S(k')\le 2(k-k'+B)-1 \text{ for $k'\le k$ }
\end{equation}
We can write an explicit formula for $S(k)$:
\begin{equation*}
  S(k) = k + 
\left\lfloor \frac k{m_1} \right\rfloor +
\left\lfloor \frac k{m_1m_2} \right\rfloor +
\cdots +
\left\lfloor \frac k{m_1m_2\ldots m_{n}} \right\rfloor
\end{equation*}
Since all $m_i\ge 2$, we get
$S(k)\le k+k/2+k/4+k/8+\cdots+k/2^n < 2k$,
proving~\eqref{left}.
For the other bound~\eqref{interval}, we use the relation
\begin{equation*}
\left\lfloor \frac k{x} \right\rfloor  -
\left\lfloor \frac {k'}{x} \right\rfloor <
\frac{k-k'}x+1
\end{equation*}
to get
\begin{equation*}
  S(k)-S(k')< (k-k')+(k-k')
\cdot(
\tfrac12+
\tfrac14+
\tfrac18+\cdots+\tfrac1{2^n}) + n < 2(k-k')+n
\end{equation*}
Since the left-hand side is an integer,
we obtain $S(k)-S(k')\le 2(k-k')+n-1$
and this implies~\eqref{interval} since
the buffer size
 $B := \lceil \frac {n}2\rceil$ satisfies $2B\ge n$.
\end{proof}

In Algorithm WORK-AHEAD, the STEPs should generate
entries $\bar\rho(1) ,\bar\rho(2),\ldots$ 
of $Q$ up to $\bar\rho(N)$, where $N := m_1m_2\ldots m_n$.

The following lemma shows that production of the STEPs
%the algorithm
may overrun their target by at most one.
Since the algorithm
 has already made provisions to generate
$\bar\rho(N)=n+1$ by extending the arrays $b$ and $m$ to size $n+1$
instead of~$n$, this one extra entry does not cause any harm.
\begin{lemma}
\label{overrun}
  In Algorithm WORK-AHEAD, 
the last entry that is added to $Q$ is
$\bar\rho(N)%=n+1
$
or $\bar\rho(N+1)%=1
$.  
\end{lemma}
\begin{proof}
  The production of
$\bar\rho(N)=n+1$ takes $n+1\ge 2B$ STEPs.
It follows that the buffer $Q$ is empty when
$\bar\rho(N)=n+1$ is inserted, regardless of whether the
 production of
$\bar\rho(N)$ is started in the first or second STEP of an iteration.

If the production of $\bar\rho(N)=n+1$ is completed in the second STEP of an iteration,
it is thus immediately consumed, which leads to termination.
If $\bar\rho(N)$ is completed in the first STEP of an iteration,
 the second STEP will produce the value  $\bar\rho(N+1)=1$,
but then the algorithm will terminate
as well.
\end{proof}

Finally, we prove the simplification of the algorithm for the binary case.

\begin{lemma}
\label{binary}
%Assume that $n\ge2$.
  In the binary version of Algorithm WORK-AHEAD, i.e., when
  $m_i=2$ for all~$i=1,\ldots,n$, the buffer $Q$ automatically never gets
more than
$B' := \max\{\lceil \frac {n+1}2\rceil,2\}$ entries,
even if the test in STEP whether the buffer is full is omitted.
\end{lemma}
\begin{proof}
  Let us assume for contradiction that
the buffer becomes overfull in  iteration~$k$,
$1\le k \le 2^n$.
This means that, before $j=\bar\rho(k)$ is
  removed from $Q$, the $2k$ STEP operations have produced more
than $k-1+B'$ values. But this is impossible, since,
as we will show,
the production of the first $k_1=k+B'$ values would have taken
\begin{equation*}
  S(k_1) = k_1 
+
\left\lfloor \frac {k_1}{2} \right\rfloor +
\left\lfloor \frac {k_1}{2^2} \right\rfloor +
\cdots +
\left\lfloor \frac {k_1}{2^n} \right\rfloor 
> 2k
\end{equation*}
STEPs. To show the last inequality,
we first consider the case $k_1<2^n$. %, $k_1= 2^n$, and $k_1> 2^n$.
%By
%Lemma~\ref{overrun}, these cases exhaust all possibilities.
%
%In the first case, 
We apply the inequality
$\lfloor x \rfloor > x-1$ and obtain
\begin{math}
  S(k_1) 
>2k_1-%\frac
{k_1}/{2^n}-n
\end{math},
and since $k_1/2^n<1$ and $S(k_1)$ is an integer, we get
\begin{equation*}
%\label{case1}
  S(k_1) 
\ge 2k_1-n
   = 2k+2B'-n > 2k.
\end{equation*}

Let us now see at what time $\bar\rho(k_1)$ 
for $k_1\ge 2^n$
 is entered into $Q$.
When $k_1= 2^n$, no round-off takes place in the formula for
$S(k_1)$, and we have
\begin{math}
%  S(k_1)=
S(2^n)=2\cdot 2^n-1
%=2k_1-1%\ge 2k_1-n
\end{math}. %, the calculation proceeds as in~\eqref{case1}.
This % In fact, the last calculation
 shows that the production of
$\bar\rho(2^n)$ is completed in the first STEP of iteration $2^n$.
In the second STEP of this iteration,
$\bar\rho(2^n+1)=1$ is added to $Q$.
Thus, when
$\bar\rho(2^n)$ is about to be retrieved,
the buffer contains $2\le B'$ elements.
Then the algorithm terminates, and no more elements are produced.
%
% This means that the only remaining case is
%  $k_1= 2^n+1$. Here we get
% \begin{math}
%   S(k_1)=S(2^n)+\bar\rho(2^n+1)=
% 2\cdot 2^n =
% 2k_1-2\ge 2k_1-n
% \end{math}, provided that $n\ge 2$, and
% the calculation proceeds as above.
% The case $n=1$ is the reason why the mimimum bound of 2 on $B'$
% was introduced, and this case can be checked directly by running the algorithm.
\end{proof}

\section{Concluding Remarks}
\label{sec:remarks}

%\hrule
%A thorough treatment about algorithms for generating $n$-tuples is given in
%Knuth~\cite[Sect.~7.2.1.1]{kn4} % Gray-code

By our approach of modeling
the Gray code in terms of a
 motion-planning game,
we were able get a mixed-radix Gray code only when all radixes $m_i$ are
odd. It remains to find a model that would work for different
even radixes or even for radixes of mixed parity.

Another motion-planning game which is related to the binary Gray code is
 the \emph{Chinese rings
  puzzle}, see Scorer, Grundy, and Smith~\cite{some-binary},
 Gardner~\cite{gardner}, or 
Knuth~\cite[pp.~285--6]{kn4}.
In this puzzle, there are at most two possible moves in every state,
like in the towers of Bucharest,
and
by simulating the Chinese rings directly, one can obtain
% thus is is easy to find a loopless way
%We mention that a direct implementation of
%leads to
 another loopless algorithm for the binary Gray code.
However, this algorithm does not seem to extend to other radixes.

\bibliography{gray}

% \section{STUFF}

% \subsection{The $m$-ary reflected Gray code}
% \label{sec:m-ary}

% \begin{tabbing}
%   \qquad\=\+ \textit{enumerate}($k$):\\
%   \qquad\=\+ \textbf{if} $k=0$:\\
%   \qquad\=\+ Visit the $n$-tuple $(a_{n-1},\ldots,a_1,a_0)$.\\
%   \textbf{return}\-\\
% \textbf{if} \textit{direction}[$k$] = UP:\\
%   \quad\=\+ 
% $a[k] := 0$ (This is redundant.)\\
% \textit{enumerate}($k-1$)\\
% $a[k] := 1$\\
% \textit{enumerate}($k-1$)\\
% $a[k] := 2$\\
% \textit{enumerate}($k-1$)\\
% \ldots\\
% $a[k] := m-1$\\
% \textit{enumerate}($k-1$)\\
%  \textit{direction}[$k$] := DOWN\-\\
% \textbf{else}: \+\\
% $a[k] := m-1$ (This is redundant.)\\
% \textit{enumerate}($k-1$)\\
% $a[k] := m-2$\\
% \textit{enumerate}($k-1$)\\
% $a[k] := m-3$\\
% \textit{enumerate}($k-1$)\\
% \ldots\\
% $a[k] := 0$\\
% \textit{enumerate}($k-1$)\\
%  \textit{direction}[$k$] := UP\-\\
% \end{tabbing}

%\newpage
%\tableofcontents
\appendix

\section{Appendix: PYTHON simulations of the algorithms}
Here we list prototype implementations in \textsc{Python}.
They should run equaly with Python~2.7 and Python~3.
The pegs, the string $a$, and the array of \texttt{directions} are
kept as global variables.
\pagebreak
\subsection{Basic procedures}
The procedure \texttt{visit} prints the string and the contents of the pegs.
\begin{verbatim}
def visit():
    print ("".join(str(x) for x in reversed(a[1:])) + " "+
           " ".join("P{}:".format(k)+",".join(map(str,p))
                                 for k,p in enumerate(pegs)))

def find_smallest_disk(exclude=None):
    list_d_k = [(p[-1],k) for k,p in enumerate(pegs) if k!=exclude and p]
    if list_d_k:
        _,k = min(list_d_k) # smallest disk not covered by D1
        return k
    return None
\end{verbatim}

\subsection{Algorithm ODD, Section~\ref{sec:odd}}

\begin{verbatim}
def initialize_m_ary_odd(m,n):
    # n-tuple of entries from the set {0,1,...,m-1}
    global pegs, a, direction
    pegs = tuple([] for k in range(m))
    for i in range(n,0,-1):
        pegs[0].append(i)
    a = (n+1)*[0] # a[0] and direction[0] is wasted
    direction = (n+1)*[+1]

def Gray_code_m_ary_odd(m):
    visit()
    while True:
        for _ in range(m-1): # repeat m-1 times:
            move_disk(m,peg=a[1])
            visit()
        k = find_smallest_disk(exclude=a[1]) # smallest disk not covered by D1
        if k==None: return
        move_disk(m,peg=k)
        visit()

def move_disk(m,peg): # move topmost disk on peg
    disk = pegs[peg].pop()
    peg     += direction[disk]
    a[disk] += direction[disk]
    if peg==m-1: direction[disk] = -1
    elif peg==0: direction[disk] = +1
    pegs[peg].append(disk)

# run the program for a test:
m=3           
initialize_m_ary_odd(m,6)
Gray_code_m_ary_odd(m)
\end{verbatim}

\subsection{Algorithm ODD-COMPRESSED, Section~\ref{odd-compressed}}
\label{imple-odd-compressed}
This implementation uses the idea of a
 ``control disk'' $D_0$ mentioned at the end of
section~\ref{sec:odd}.
For uniformiy, we attach a direction also to
the dial positions $0$ and $m-1$,
namely the direction in with the next move will proceed
(in contrast to 
the convention~\eqref{eq:up-down} used in
 Section~\ref{odd-compressed}).
\begin{verbatim}
def initialize_m_ary_odd_compressed(m,n):
    # n-tuple of entries from the set {0,1,...,m-1}
    global pegs, a
    pegs = tuple([] for k in range(3))
    for i in range(n,0,-1):
        pegs[0].append((i,0,+1))
    a = (n+1)*[0] # a[0] is wasted

def Gray_code_m_ary_odd_compressed(m):
    visit()
    while True:
        for control_disk in (2,0):
            k = find_smallest_disk(exclude=control_disk)
            if k==None: return # All disks are on the same peg.
            move_disk_compressed(m,peg=k)
            visit()

def move_disk_compressed(m,peg):
    # retrieve topmost disk on peg:
    disk,value,direction = pegs[peg].pop()
    if value in (0,m-1) or value+direction in (0,m-1):
        peg += direction
    value += direction
    a[disk] += direction
    if value in (0,m-1):
        direction = -direction
    pegs[peg].append((disk,value,direction))

# run the program for a test:
m=5
initialize_m_ary_odd_compressed(m,5)
Gray_code_m_ary_odd_compressed(m)
\end{verbatim}

\vfill\pagebreak%[3]

\subsection{Algorithm EVEN, Section~\ref{even}}
%\label{imple-odd-compressed}

\begin{verbatim}
def initialize_m_ary_even(m,n):
    initialize_m_ary_odd(m+1,n) # use m+1 pegs

def Gray_code_m_ary_even(m):
    peg_disk1 = 0 # position of disk D1
    visit()
    while True:
        for _ in range(m-1): # repeat m-1 times:
            turn_disk(m,peg=peg_disk1)
            peg_disk1 = (peg_disk1+1) % (m+1)
            visit()
        k = find_smallest_disk(exclude=peg_disk1)
             # smallest disk not covered by D1
        if k==None: return
        turn_disk(m,peg=k)
        visit()

def turn_disk(m,peg): # move the topmost disk on peg
    disk = pegs[peg].pop()
    if disk%2==1:
        peg = (peg + 1) % (m+1)
    else:
        peg = (peg - 1) % (m+1)
    pegs[peg].append(disk)
    a[disk] += direction[disk]
    if a[disk]==m-1: direction[disk] = -1
    elif a[disk]==0: direction[disk] = +1
            
# run the program for a test:
m=4
initialize_m_ary_even(m,6)
Gray_code_m_ary_even(m)
\end{verbatim}

\subsection{Truly loopless implementation of Algorithm EVEN, Section~\ref{simulation}}

A peg is a list of pairs $(a,b)$ with $a\le b$, denoting a maximal
subset $a,a+1,\ldots,b$ of consecutive disks (an interval).
The pairs are sorted, with the smallest disks at the end (at the
``top'').

\begin{verbatim}
def initialize_m_ary_even_intervals(m,n):
    global pegs,position
    initialize_m_ary_even(m,n)
    pegs[0][:]=[(1,n)]
    position=[0]*(n+1)

\end{verbatim}

\vskip 1pt plus 5cm
\goodbreak

\begin{verbatim}
def Gray_code_m_ary_even_intervals(m,n):
    visit()
    while True:
        for _ in range(m-1): # repeat m-1 times:
            turn_disk_intervals(m,peg=position[1])
            visit()
        d = find_smallest_missing_disk(position[1])
             # smallest disk not covered by D1
        if d>n: return
        turn_disk_intervals(m,peg=position[d])
        visit()

def turn_disk_intervals(m,peg): # move topmost disk on peg
    disk,d2 = pegs[peg][-1]
    # remove disk from peg:
    if disk==d2:
        pegs[peg].pop()
    else:
        pegs[peg][-1]=(disk+1,d2)
    if disk%2==1:
        peg = (peg + 1) % (m+1) # turn the disk "clockwise"

    else:
        peg = (peg - 1) % (m+1) # turn the disk "counterclockwise"
    position[disk]=peg
    # add disk to peg:
    if pegs[peg]:
        d1,d2 = pegs[peg][-1]
        if disk<d1-1:
            pegs[peg].append((disk,disk))
        else:
            pegs[peg][-1]=(disk,d2)
    else:
        pegs[peg].append((disk,disk))
    a[disk] += direction[disk]
    if a[disk]==m-1: direction[disk] = -1
    elif a[disk]==0: direction[disk] = +1

def find_smallest_missing_disk(peg):
    (_,d2)=pegs[peg][-1]
    return d2+1

# run the program for a test:
m=4
initialize_m_ary_even_intervals(m,6)
Gray_code_m_ary_even_intervals(m,6)
\end{verbatim}

\pagebreak

\subsection{Algorithm WORK-AHEAD, Section~\ref{sec:ahead}}

\begin{verbatim}
def STEP():
    global j, b,m,B,Q
    if b[j]==m[j]-1:
        b[j]=0
        j += 1
    else:
        if len(Q)<B:
            b[j] += 1
            Q.append(j)
            j = 1

def initialize_work_ahead(n):
    global a,b,direction,B,Q,j
    a = (n+1)*[0] # a[0], b[0], and direction[0] is wasted
    direction = (n+1)*[+1]
    b = (n+2)*[0]
    from collections import deque
    B = (n+1)//2
    Q = deque()
    j = 1

def Gray_code_work_ahead(n):
    while True:
        VISIT()
        STEP()
        STEP()
        j = Q.popleft()
        if j==n+1: break
        a[j] += direction[j]
        if a[j] in (0,m[j]-1): direction[j] *= -1

def VISIT():
    print ("".join(str(x) for x in reversed(a[1:])))
        
# run the program for a test:
n=4
m=[0]+[2,4,5,2]+[2] # initial 0 and final 2 are artificial
initialize_work_ahead(n)
Gray_code_work_ahead(n)
\end{verbatim}

\goodbreak
\vfill\pagebreak

\subsection{Algorithm WORK-AHEAD for the binary Gray code, Section~\ref{sec:ahead}}

\begin{verbatim}
def STEP_binary():
    global j, b,B,Q
    if b[j]==1:
        b[j]=0
        j += 1
    else:
        if len(Q)>=B:
            print ("error")
            exit(1)
        b[j]=1
        Q.append(j)
        j = 1

def initialize_work_ahead_binary(n):
    global a,b,B,Q,j
    a = (n+1)*[0] # a[0], b[0], and direction[0] is wasted
    b = (n+2)*[0]
    from collections import deque
    B = max(2,(n+2)//2)
    Q = deque()
    j = 1

def Gray_code_work_ahead_binary(n):
    while True:
        VISIT()
        STEP_binary()
        STEP_binary()
        j = Q.popleft()
        if j==n+1: break
        a[j] = 1-a[j]
        
# run the program for a test:
n=4
initialize_work_ahead_binary(n)
Gray_code_work_ahead_binary(n)
\end{verbatim}

\goodbreak
\vfill\pagebreak

\subsection{Algorithm WORK-AHEAD with the modification of Section~\ref{sec:alternative-STEP}}

\begin{verbatim}
def STEP_x():
    global j, b,m,B,Q,n
    if j==n+1: raise StopIteration
    if b[j]==m[j]-1:
        b[j]=0
        j += 1
    else:
        if len(Q)<B:
            b[j] += 1
            Q.append(j)
            j = 1

def initialize_work_ahead_x():
    global n,a,b,direction,B,Q,j
    a = (n+1)*[0] # a[0], b[0], and direction[0] is wasted
    direction = (n+1)*[+1]
    b = (n+1)*[0]
    from collections import deque
    B = (n+1)//2
    Q = deque()
    j = 1
    
def Gray_code_work_ahead_x():
    try:
        while True:
            VISIT()
            STEP_x()
            STEP_x()
            j = Q.popleft()
            a[j] += direction[j]
            if a[j] in (0,m[j]-1): direction[j] *= -1
    except StopIteration:
        return

# run the program for a test:
n=4
m=[0]+[2,4,5,2] # initial entry 0 is redundant
initialize_work_ahead_x()
Gray_code_work_ahead_x()
\end{verbatim}

\vfill\pagebreak
\subsection{General mixed-radix Gray code generation according to the recursive definition of Section~\ref{general}}

In order to have a reference implementation for comparing the results,
we give a program straight from the definition~\eqref{m-ary}
 of Section~\ref{general}, extended to arbitrary mixed radices $(m_1,\ldots,m_n)$.
When this program is run with the data specified below, it should
produce the same Gray codes as all the previous example programs
combined (apart from the additional state of the pegs that is reported
by these programs).
The outputs coincide precisely after
stripping everything after the first blank on each line.
The source files of this arXiv preprint include scripts that will
 extract the program code from the \LaTeX\ file of this appendix
(\texttt{extract-code-from-appendix.py})
and compare the outputs to check the results
(\texttt{check-output.sh}) after running the examples,
\begin{verbatim}
def mixed_Gray_code(ms):
    "a generator for the mixed-radix Gray code"
    if ms:
        m = ms[0]
        G1 = mixed_Gray_code(ms[1:])
        while True:
            g = next(G1)
            for lastdigit in range(0,m):
                yield g+(lastdigit,)
            g = next(G1)
            for lastdigit in reversed(range(0,m)):
                yield g+(lastdigit,)
    else:
        yield ()

for result in (
        mixed_Gray_code([3]*6),
        mixed_Gray_code([5]*5),
        mixed_Gray_code([4]*6),
        mixed_Gray_code([4]*6),
        mixed_Gray_code([2,4,5,2]),
        mixed_Gray_code([2]*4),
        mixed_Gray_code([2,4,5,2]),
        ):
    print ("##########")
    for g in result:
        print ("".join(str(x) for x in g))
\end{verbatim}

% The following GNU/Unix script will split the output of the programs,
% which is assumed be be in the file output.txt, into the part
% from the reference implementation in the last section
% (output-ref.txt),
% and the parts from the different algorithms before (output1.txt),
% purge them from additional information beyond the pure Gray codes,
% and compare them:

% % sed -e '/recursive definition/Q;/^$/d;/subsection/d;s/ .*//' output.txt >output1.txt
% % sed -e '1,/recursive definition/d;/^$/d;' output.txt >output-ref.txt
% % diff output1.txt output-ref.txt

% Local Variables: 
% mode: latex
% mode: TeX-PDF
% TeX-master: "p19-herter-arx"
% End:

%\input appendix-herter

\end{document}